\pdfoutput=1 

\documentclass[letterpaper, 10 pt, conference]{ieeeconf}  % Comment this line out if you need a4paper

\usepackage{amsmath,amssymb,amsfonts}
\usepackage{graphicx}
\usepackage{cite}
\usepackage{algorithmic}
\usepackage{ntheorem}
\usepackage{scalefnt}
\usepackage[ruled,vlined]{algorithm2e}

\usepackage[bookmarks]{hyperref}
\hypersetup{colorlinks = true, citecolor = black, linkcolor = black, urlcolor = black}       % hyperlinks

\IEEEoverridecommandlockouts                              % This command is only needed if 
\title{\LARGE \bf
	Cooperative Path Integral Control for Stochastic Multi-Agent Systems
}

\author{Neng Wan$^{1}$, Aditya Gahlawat$^{1}$, Naira Hovakimyan$^{1}$, Evangelos A. Theodorou$^{2}$, and Petros G. Voulgaris$^{3}$% <-this % stops a space
	\thanks{*This work is supported in part by AFOSR under Grant FA9550-15-1-0518, NSF under Grant CCF-1717154, CMMI-1663460, CNS-1932288, CNS-1932529, ECCS-1739732, NRI ECCS-1830639, and ZJU-UIUC Institute Research Program. }% <-this % stops a space
	\thanks{$^{1}$Neng Wan, Aditya Gahlawat, and Naira Hovakimyan are with the Department of Mechanical Science and Engineering, University of Illinois at Urbana-Champaign, Urbana, IL 61801, USA. 
		{\tt\footnotesize \{nengwan2, gahlawat, nhovakim\}@illinois.edu}}%
	\thanks{$^{2}$Evangelos A. Theodorou is with the Department of Aerospace Engineering, Georgia Institute of Technology, Atlanta, GA 30332, USA.
		{\tt\footnotesize evangelos.theodorou@gatech.edu}}
	\thanks{$^{3}$Petros G. Voulgaris is with the Department of Mechanical Engineering, University of Nevada, Reno, NV 89557, USA.
		{\tt\footnotesize pvoulgaris@unr.edu}}%
}

\theoremheaderfont{\bfseries}
\theorembodyfont{\normalfont}

\newtheorem{theorem}{Theorem}
\newtheorem{proposition}[theorem]{Proposition}

\renewenvironment{proof}{{\textit{Proof.}}}

\begin{document}
	\maketitle
	\thispagestyle{empty}
	\pagestyle{empty}

	%%%%%%%%%%%%%%%%%%%%%%%%%%%%%%%%%%%%%%%%%%%%%%%%%%%%%%%%%%%%%%%%%%%%%%%%%%%%%%%%
	\begin{abstract}
		A distributed stochastic optimal control solution is presented for cooperative multi-agent systems. The network of agents is partitioned into multiple factorial subsystems, each of which consists of a central agent and neighboring agents. Local control actions that rely only on agents' local observations are designed to optimize the joint cost functions of subsystems. When solving for the local control actions, the joint optimality equation for each subsystem is cast as a linear partial differential equation and solved using the Feynman-Kac formula. The solution and the optimal control action are then formulated as path integrals and approximated by a Monte-Carlo method. Numerical verification is provided through a simulation example consisting of a team of cooperative UAVs.
	\end{abstract}

	%%%%%%%%%%%%%%%%%%%%%%%%%%%%%%%%%%%%%%%%%%%%%%%%%%%%%%%%%%%%%%%%%%%%%%%%%%%%%%%%
	\section{Introduction}
	
	The research on control and planning in multi-agent systems (MASs) has been developing rapidly during the last decade with increasing demand from areas, such as cooperative vehicles~\cite{Cichella_CSM_2016}, Internet of Things~\cite{Ota_TSG_2012} and intelligent infrastructures~\cite{Dorfler_TCNS_2016}. Distinct from other control problems, control of MASs is characterized by challenges of limited information and resources of local agents, randomness of agent dynamics and communication networks, and optimality and robustness of joint performance. A good summary of recent progress in multi-agent control can be found in~\cite{Amato_CDC_2013, Cao_TII_2013, Frank_2013, Oh_Auto_2015, Qin_TIE_2017}. Building upon these results, this paper proposes a cooperative optimal control scheme by extending the path integral control (PIC) algorithm in~\cite{Theodorou_JMLR_2010} for a general type of stochastic MAS subject to limited feedback information and computational resource.
	
	Path integral control is a model-based stochastic optimal control (SOC) algorithm that linearizes and solves the stochastic Hamilton-Jacobi-Bellman (HJB) equation with the facilitation of Cole-Hopf transformation, \textit{i.e.} exponential transformation of value function~\cite{Fleming_AMO_1977, Todorov_NIPS_2007}. Compared with other continuous-time SOC techniques, since PIC formulates the optimality equations in linear form, it enjoys the superiority of closed-form solution~\cite{Pan_NIPS_2015} and superposition principle~\cite{Todorov_NIPS_2009}, which makes PIC a popular control scheme for robotics~\cite{Williams_TRO_2018}. Some recent developments of single-agent PIC algorithms can be found in~\cite{Theodorou_JMLR_2010, Gomez_KDD_2014, Pan_NIPS_2015, Williams_JGCD_2017}.
	
	Different from many prevailing distributed control algorithms~\cite{Cao_TII_2013, Frank_2013, Oh_Auto_2015, Qin_TIE_2017}, such as consensus and synchronization that usually assume a given behavior, multi-agent SOC allows agents to have different objectives and optimizes the action choices for more general scenarios~\cite{Amato_CDC_2013}. Nevertheless, it is not straightforward to extend the single-agent SOC or PIC algorithms to MASs. The exponential growth of dimensionality in MASs and the consequent surges in computation and data storage demand more sophisticated and preferably distributed planning and execution algorithms. The involvement of communication networks (and constraints) requires the multi-agent SOC algorithms to achieve stability and optimality subject to local observation and more involved cost function. While a few efforts have been made on multi-agent PIC, most of these control algorithms still depend on the knowledge of the global state, \textit{i.e.} a fully connected communication network, which may not be feasible or affordable to attain in practice. Some multi-agent PIC algorithms also assume that the joint cost function can be factorized over agents, which simplifies the multi-agent control problem into multiple single-agent problems by ignoring the correlations among agents, and some features and advantages of MASs are therefore forfeited. Broek \textit{et al.} investigated the multi-agent PIC problem for continuous-time systems governed by It\^o diffusion process~\cite{Broek_2008_JAIR}; a path integral formula was put forward to approximate the optimal control actions, and a graphical model inference approach was adopted to predict the optimal path distribution; nonetheless, the optimal control policy assumed an accurate and complete knowledge of global state, and the inference was conducted on the basis of mean field approximation, which assumes that the cost function can be disjointly factorized over agents. A distributed PIC algorithm with infinite-horizon and discounted cost was applied to solving a distance-based formation problem for nonholonomic vehicular network without explicit communication topology in~\cite{Anderson_Robotica_2014}. Cooperative PIC problem was also recently studied in~\cite{Williams_JGCD_2017} as an accessory result for a novel single-agent PIC algorithm; an augmented dynamics was built by piling up the dynamics of all agents, and a single-agent PIC algorithm was then applied to the augmented system. Nonetheless, the results resorting to augmented dynamics presume  fully connected network and face the challenge that the computational and sampling schemes that originated from single-agent problem may become inefficient and possibly fail as the dimensions of augmented state and control grow exponentially in the number of agents.
	
	%5) Main results and contributions of this paper. 
	In order to address the aforementioned issues, cooperative PIC algorithm is investigated in this paper with consideration of local observation, joint cost function, and an efficient computational method. A distributed control framework that partitions the connected communication network into multiple factorial subsystems is proposed, and the local PIC protocol of each individual agent is computed in a subsystem, which consists of an interested central agent and its neighboring agents. Under this framework, every (central) agent relying on the local observation acts optimally to minimize a joint cost function of its subsystem, and the complexities of computation and sampling are now related to the size (amount of agents) of each factorial subsystem instead of the entire network. When solving for the local optimal control action, instead of adopting the mean-field approximation and factorizing the cost function over individual agents, joint cost functions are considered inside the subsystems. The joint optimality equation of each subsystem is first cast into a joint stochastic HJB equation and formulated as a linear partial differential equation (PDE) that can be solved by the Feynman-Kac lemma. The solution of optimality equation and joint optimal control action are then formulated as path integrals and approximated by Monte-Carlo (MC) method. Parallel random sampling technique is introduced to accelerate and parallelize the approximation of the PIC solutions, and state measurements and sampled trajectory data are exchanged between neighboring agents. Illustrative examples of a cooperative UAV team are presented to verify the effectiveness and advantages of cooperative PIC algorithm.
	
	This paper is organized as follows: \hyperref[sec2]{Section II} formulates the control problem; \hyperref[sec3]{Section III} investigates the cooperative PIC algorithm; \hyperref[sec4]{Section~IV} presents the simulation example, and \hyperref[sec5]{Section V} draws the conclusions. For a matrix $X$ and a vector $v$, $|X|$ denotes the determinant of $X$, and $\|v\|^2_M = v^\top M v$. For a set $S$, $|S|$ denotes its cardinality.

	\section{Problem Formulation}\label{sec2}
	The mathematical representation for MAS, the cooperative control framework, including the stochastic dynamics and optimal control formulation are  introduced in this section. 
	
	\subsection{Multi-Agent System and Factorial Subsystems}
	For a MAS with $N$ homogeneous agents indexed by $\mathcal{D} = \{1,2,\allowbreak\cdots,  N\}$, we use a connected graph $\mathcal{G} = \{ \mathcal{V}, \mathcal{E} \}$ to represent the bilateral communication network underlying this MAS, where vertex $v_i \in \mathcal{V}$ denotes agent $i$ and undirected edge $(v_i, v_j) \in \mathcal{E}$, if agent $i$ and $j$ can directly exchange their information with each other. For agent $i$, $\mathcal{N}_i$ is the index set of all agents neighboring  agent $i$. The factorial subsystem $\bar{\mathcal{N}}_i = \mathcal{N}_i \cup \{i\}$ comprises a central agent $i$ and its neighboring agents $\mathcal{N}_i$. \hyperref[fig1]{Figure~1} shows an example of a MAS and two of its factorial subsystems.
	
	Our cooperative control framework is a trade-off scheme between the cooperation among agents and computational complexity. Instead of optimizing a fully factorized cost function based on the mean-field assumption~\cite{Broek_2008_JAIR} or a global cost function relying on the knowledge of global states~\cite{Williams_JGCD_2017}, we design the local control action $u_i$, which depends on the local observation of agent $i \in \mathcal{D}$, by minimizing the joint cost function of subsystem $\mathcal{\bar N}_i$. Under this cooperative control scheme, we not only capture and optimize the correlation between neighboring agents, but circumvent the dependency on the global state as well as the exponential growth of the global state dimension.	
	\begin{figure}[htpb!]\label{fig1}\vspace{-0.5em}
		\centerline{\includegraphics[width=0.75\columnwidth]{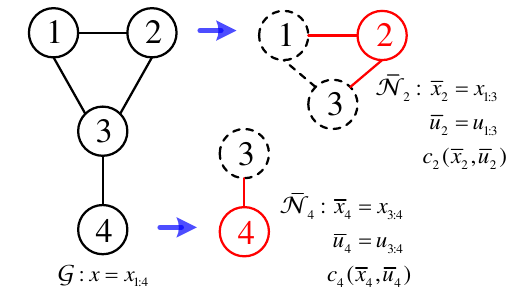}}
		\caption{MAS $\mathcal{G}$ and some of its factorial subsystems.}
	\end{figure}\vspace{-1em}
	
	\subsection{Stochastic Optimal Control}
	Consider a network of agents with homogeneous and mutually independent passive dynamics. We use the following It\^{o} diffusion process to describe the joint dynamics of subsystem $\bar{\mathcal{N}}_i$
		\begin{equation}\label{dynamics}
		d \bar{x}_i = \bar{f}_i(\bar{x}_i,t) dt + \bar{B}_i(\bar{x}_i)  \left[ \bar{u}_i(\bar{x}_i,t) dt + \bar{\sigma}_i d\bar{w}_i  \right],
		\end{equation}
		where $\bar{x}_i = [x_i^\top, x_{j\in\mathcal{N}_i}^\top] \in \mathbb{R}^{M\cdot |\bar{\mathcal{N}}_i|}$ is the joint states of subsystem $\mathcal{\bar{N}}_i$ with $M$ denoting the state dimension of a single agent; $\bar{f}_i(\bar{x}_i,t) = [f_i(x_i, t)^\top, f_{j\in\mathcal{N}_i}(x_j, t)^\top] \in \mathbb{R}^{M\cdot |\bar{\mathcal{N}}_i|}$, $\bar{B}_i(\bar{x}_i) = \textrm{diag}\{B_i(x_i), B_{j\in \mathcal{N}_i}(x_j) \} \in \mathbb{R}^{M\cdot |\bar{\mathcal{N}}_i| \times P \cdot |\bar{\mathcal{N}}_i|}$, and $\bar{u}_i(\bar{x}_i,t)  =  [u_i(\bar{x}_i, t)^\top, u_{j\in\mathcal{N}_i}(\bar{x}_i, t)^\top]^\top  \in \mathbb{R}^{P \cdot |\bar{\mathcal{N}}_i|}$ are the joint passive dynamics, joint control matrix, and joint control input of subsystem $\mathcal{\bar{N}}_i$ with $P$ being the input dimension of a single agent; joint noise $d\bar{w}_i\in \mathbb{R}^{P \cdot |\bar{\mathcal{N}}_i|}$ is a vector of Brownian components with zero mean and unit rate of variance, and $\bar{\sigma}_i  \in \break \mathbb{R}^{P \cdot |\bar{\mathcal{N}}_i| \times P \cdot |\bar{\mathcal{N}}_i|}$ is  a positive semi-definite matrix that denotes the joint covariance of noise $d\bar{w}_i$. By rearranging the components in $\bar{x}_i$, the stochastic dynamics~\eqref{dynamics} can be partitioned into a non-directly actuated part with states $\bar{x}_{i(n)} \in \mathbb{R}^{U\cdot |\mathcal{\bar{N}}_i|}$ and a directly actuated part with states $\bar{x}_{i(d)} \in \mathbb{R}^{D\cdot |\mathcal{\bar{N}}_i|}$, where $U$ and $D$ respectively denote the dimensions of non-directly actuated states and directly actuated states for a single agent. Consequently, the joint passive dynamics and control matrix can respectively be partitioned as $\bar{f}_i(\bar{x}_i, t) = [\bar{f}_{i(n)}(\bar{x}_i, t)^\top, \bar{f}_{i(d)}(\bar{x}_i, t)^\top]^\top$ and $\bar{B}_i(\bar{x}_i) = [0, \bar{B}_{i(d)}(\bar{x}_i)]^\top$.

		Let $\mathcal{\bar I}_i$ denote the set of joint interior states in subsystem $\bar{\mathcal{N}}_i$. When $\bar{x}_i \in \mathcal{\bar I}_i$, the joint running cost function of $\bar{\mathcal{N}}_i$ is defined as
		\begin{equation}\label{Run_Cost}
		c_i(\bar{x}_i, \bar{u}_i) = q_i(\bar{x}_i) + \frac{1}{2}\bar{u}_i(\bar{x}_i, t)^\top \bar{R}_i \bar{u}_i(\bar{x}_i, t),
		\end{equation}
		where $q_i(\bar{x}_i) \geq 0$ is a state-related cost, and $\bar{u}_i(\bar{x}_i, t)^\top \cdot \bar{R}_i  \bar{u}_i(\bar{x}_i, t)$ is a control-quadratic term with $\bar{R}_i = \textrm{diag}\{R_i, \break R_{j\in \mathcal{N}_i}\} \in \mathbb{R}^{P \cdot |\bar{\mathcal{N}}_i| \times P \cdot |\bar{\mathcal{N}}_i|}$ being positive definite. Let $\mathcal{\bar B}_i$ denote the set of joint exit states in subsystem $\bar{\mathcal{N}}_i$. When $\bar{x}_i^{t_f} \in \mathcal{\bar B}_i$, the terminal cost function is defined as $\phi_i(\bar{x}_i^{t_f})$, where $t_f$ indicates the exit time. A cost-to-go function $J^{\bar{u}_i}_i(\bar{x}_i^t, t) $ for first-exit problem subject to control policy $\bar{u}_i$ can then be defined as $J^{\bar{u}_i}_i(\bar{x}_i^t, t) = \mathbb{E}^{\bar{u}_i}_{\bar{x}_i^t, t} [ \phi_i(\bar{x}_{i}^{t_f}) + \int_{t}^{t_f} c_i(\bar{x}_i(\tau), \bar{u}_i(\tau)) \ d\tau]$. We can minimize $J^{\bar{u}_i}_i(\bar{x}_i^t, t)$ by solving for the joint optimal control action $\bar{u}_i^*$ from the following joint optimality equation
		\begin{equation}\label{optimal_eq}
		V_i(\bar{x}_i, t) = \min_{\bar{u}_i} \mathbb{E}^{\bar{u}_i}_{\bar{x}_i^t, t} \left[ \phi_i(\bar{x}_{i}^{t_f}) + \int_{t}^{t_f} c_i(\bar{x}_i(\tau), \bar{u}_i(\tau)) \ d\tau  \right],
		\end{equation}
		where the value function $V_i(\bar{x}_i, t)$ is the expected cumulative running cost for starting at $\bar{x}_i$ and acting optimally thereafter. By following the local optimal control action $u_i^*$ marginalized from $\bar{u}_i^*$, each (central) agent acts optimally to minimize $J^{\bar{u}_i}_i(\bar{x}_i^t, t)$, while the (global) optimality condition in~\eqref{optimal_eq} can only be attained when $\mathcal{G}$ is fully connected, since the local optimal control action $u_j^*$ of (neighboring) agent $j\in\mathcal{N}_i$ usually does not accord with $\bar{u}^*_i$. This conflict widely exists in distributed optimal control and optimization problems when the networks are subject to local/partial observation and limited communication, and some serious and heuristic studies on the global- and sub-optimality of distributed systems can be found in \cite{Nedic_TAC_2009, Frank_2013, Johari_MOR_2004, Voulgaris_CDC_2017} and references therein. We will not dive into this technical detail, as the objective of this paper is to propose a sub-optimal PIC scheme with sufficient computation and sample efficiency in networked MAS.

	\section{Cooperative Path Integral Control}\label{sec3}
	The joint optimality equation~\eqref{optimal_eq} is first cast as a linear PDE that can be solved by the Feynman-Kac formula. The solution and the joint optimal control action $\bar{u}_i^*$ are then formulated as path integrals that can be approximated by distributed stochastic sampling. 
	
	\subsection{Linear Formulation}
	We first formulate the joint optimality equation~\eqref{optimal_eq} as a joint HJB equation and cast it as a linear PDE with the following exponential transformation, which is also known as the Cole-Hopf transformation~\cite{Fleming_AMO_1977}:
	\begin{equation}\label{CH_Trans}
	Z(\bar{x}_i, t) = \exp[-V_i(\bar{x}_i, t) / \lambda_i],
	\end{equation}
	where $\lambda_i\in\mathbb{R}$ is a scalar, and $Z(\bar{x}_i, t)$ is the desirability function of joint state $\bar{x}_i$ at time $t$. The following theorem states the joint HJB equation, the joint optimal control action $\bar{u}_i^*$, and a linear formulation for~\eqref{optimal_eq} along with a closed-form solution.  
	\begin{theorem}\label{thm1}
		For the factorial subsystem $\mathcal{\bar{N}}_i$ subject to joint dynamics~\eqref{dynamics} and running cost function~\eqref{Run_Cost}, the joint optimality equation~\eqref{optimal_eq} is equivalent to the following joint stochastic HJB equation 
		\begin{align}\label{sto_HJB}
		&\hspace{-8pt}-\partial_t V_i(\bar{x}_i, t) = \min_{\bar{u}_i}  \mathbb{E}_{\bar{x}_i,t}^{\bar{u}_i} \bigg[ \frac{1}{2} \bar{u}_i(\bar{x}_i, t)^\top\bar{R}_i\bar{u}_i(\bar{x}_i, t) + q_i(\bar{x}_i, t)  \allowdisplaybreaks \nonumber \\
		& \hspace{10pt} + \sum_{j \in \bar{\mathcal{N}}_i} [f_j(x_j,t)  + B_j(x_j) u_j(\bar{x}_i, t)]^\top  \nabla_{x_j} V_i(\bar{x}_i,t)    \allowdisplaybreaks   \nonumber \\
		& \hspace{10pt} + \frac{1}{2} \sum_{j \in \bar{\mathcal{N}}_i}  \textrm{tr} \left( B_j(x_j)\sigma_j \sigma_j^\top B_j(x_j)^\top  \nabla_{x_jx_j}V_i(\bar{x}_i, t) \right) \allowdisplaybreaks \bigg]
		\end{align}
		with boundary condition $V_i(\bar{x}_i, t_f) = \phi_i(\bar{x}_i)$. The minimum of~\eqref{sto_HJB} is attained by the joint optimal control action
		\begin{equation}\label{OptimalAction}
		\bar{u}^*_i(\bar{x}_i, t) = -\bar{R}_i^{-1}\bar{B}_i(\bar{x}_i)^\top \nabla_{\bar{x}_i}V_i(\bar{x}_i, t).
		\end{equation}
		With transformation~\eqref{CH_Trans}, control action~\eqref{OptimalAction} and condition $\bar{R}_i = (\bar{\sigma}_i\bar{\sigma}_i^\top / \allowbreak \lambda_i)^{-1}$, the joint stochastic HJB equation~\eqref{sto_HJB} can be formulated as 
		\begin{align}\label{Linear_PDE}
		&\partial_{t} Z_i(\bar{x}_i, t) = \bigg[\frac{q_i(\bar{x}_i, t)}{\lambda_i} - \sum_{j\in\bar{\mathcal{N}}_i} f_j(x_j, t)^\top  \nabla_{x_j}  \\
		& \hspace{30pt}- \frac{1}{2}\sum_{j \in \bar{\mathcal{N}}_i} \textrm{tr}\left( B_j(x_j) \sigma_j \sigma_j^\top B_j(x_j)^\top  \nabla_{x_jx_j} \right)  \bigg]Z_i(\bar{x}_i,t) \nonumber
		\end{align}
		with boundary condition $Z_i(\bar{x}_i, t_f) = \exp[- \phi_i(\bar{x}_i) / \lambda_i]$ and has a closed-form solution
		\begin{equation}\label{Z_Solution}
		Z_i(\bar{x}_i, t) = \mathbb{E}_{\bar{x}_i,t}\left[ \exp\left( -\frac{\phi_i(\bar{y}^{t_f}_i)}{\lambda_i}  - \int_{t}^{t_f}   \frac{q_i(\bar{y}_i, \tau)}{\lambda_i} \ d\tau \right) \right],
		\end{equation}
		where the diffusion process $\bar{y}(t)$ is subject to $d \bar{y}_i(\tau) = \bar{f}_i(\bar{y}_i,\tau) d\tau + \bar{B}_i(\bar{y}_i)   \bar{\sigma}_i  \cdot d\bar{w}_i(\tau)$ initiated at $\bar{y}_i(t) = \bar{x}_i(t)$.
	\end{theorem}
	\begin{proof}
		See \hyperref[AppI]{Appendix~I} for the proof.
	\end{proof}
	
	The condition $\bar{R}_i = (\bar{\sigma}_i\bar{\sigma}_i^\top / \allowbreak \lambda_i)^{-1}$ implies that high control cost is assigned to a control channel with low variance noise, while a control channel with high variance noise has cheap control cost. With some auxiliary techniques this condition can be relaxed~\cite{Stulp_ICML_2012}.

	\subsection{Path Integral Approximation}
	While a closed-form solution for $Z_i(\bar{x}_i, t)$ is given in~\hyperref[thm1]{Theorem~1}, the expectation over all uncontrolled trajectories initiated at $(\bar{x}_i, t)$ is intractable to compute. A conventional approach in statistical physics for evaluating this expectation is to rewrite it as a path integral and approximate the integral with sampling methods. The following proposition gives the path integral formulae for the desirability function $Z_i(\bar{x}_i,t)$ and the joint optimal control action $\bar{u}^*_i(\bar{x}_i, t)$.

	\begin{proposition}\label{prop3} 
		Divide the time span from $t$ to $t_f$ into $K$ intervals of even length $\varepsilon > 0$, $t = t_0 < t_1 < \cdots < t_K = t_f$, and let $\bar{x}_i^{(k)}  = [\bar{x}_{i(n)}^{(k)\top}, \bar{x}_{i(d)}^{(k)\top}]^\top$ denote the joint trajectory segment on time interval $[t_{k-1}, t_k)$ subject to joint dynamics~\eqref{dynamics}, zero control $\bar{u}_i = 0$, and initial condition $\bar{x}_i(t) = \bar{x}_i^{(0)}$. The desirability function~\eqref{Z_Solution} can then be reformulated as a path integral
		\begin{align}
		&Z_i(\bar{x}_i,t) =  \lim_{\varepsilon \downarrow 0}  \int   \exp\Big(  -\tilde{S}_i^{\varepsilon, \lambda_i}(\bar{x}_i^{(0)}, \bar{\ell}_i, t_0)  \\
		& \hspace{98pt}- K D |\mathcal{\bar N}_i| / 2 \cdot \log (2\pi \varepsilon) \Big)  d\bar{\ell}_i, \nonumber
		\end{align}	
		where the path variable $\bar \ell_i = (\bar{x}^{(1)}_i, \cdots, \bar{x}^{(K)}_i)$ represents all uncontrolled trajectories of subsystem $\mathcal{\bar N}_i$ starting at $(\bar{x}_i, t)$, and the generalized path value
		\begin{align}\label{Gen_Path_Value}
		& \tilde{S}_i^{\varepsilon, \lambda_i}(\bar{x}_i^{(0)}, \bar{\ell}_i, t_0) = \frac{\phi_i(\bar{x}^{(K)}_i)}{\lambda_i} + \frac{\varepsilon}{\lambda_i}\sum_{k=0}^{K-1}q_i(\bar{x}^{(k)}_i, t_k)  \allowdisplaybreaks\\
		&\hspace{25pt} + \frac{\varepsilon}{2\lambda_i} \sum_{k=0}^{K-1}  \left\|  \alpha_i^{(k)} 	\right\|^2_{\left(H_i^{(k)}\right)^{-1}} +   \frac{1}{2}\sum_{k=0}^{K-1} \log \left|H_i^{(k)}\right| \nonumber
		\end{align}	
		with $\alpha_i^{(k)} = (\bar{x}_{i(d)}^{(k+1)} - \bar{x}_{i(d)}^{(k)}) / \varepsilon -\bar{f}_{i(d)}(\bar{x}^{(k)}_i, t_k)$ and  $H_i^{(k)}  =  \lambda_i   \bar{B}_{i(d)}(\bar{x}^{(k)}_i)\bar{R}_i^{-1} \bar{B}_{i(d)}(\bar{x}^{(k)}_i)^\top = \bar{B}_{i(d)}(\bar{x}^{(k)}_i)  \bar{\sigma}_i\bar{\sigma}_i^\top \cdot \bar{B}_{i(d)}(\bar{x}^{(k)}_i)^\top$. Hence, the joint optimal control action for subsystem $\bar{\mathcal{N}}_i$ can then be reformulated as a path integral
		\begin{align}\label{OptCtrlPath}
		&\bar{u}^*_i(\bar{x}_i, t) = \bar\sigma_i \bar\sigma_i^\top \bar{B}_{i(d)}(x_i)^\top \cdot \lim_{\varepsilon \downarrow 0} \int  \tilde{p}^*_i(\bar{\ell}_i | \bar{x}_i^{(0)}, t_0) \\
		&\hspace{150pt}\times \tilde{u}_i(\bar{x}_i^{(0)}, \bar{\ell}_i, t_0) \ d\bar{\ell}_i,\nonumber
		\end{align}	
		where the optimal path distribution is
		\begin{equation}\label{OptPathDist}
		\tilde{p}^*_i(\bar{\ell}_i | \bar{x}_i^{(0)}, t_0) = \frac{\exp( -\tilde{S}_i^{\varepsilon, \lambda_i}(\bar{x}_i^{(0)}, \bar{\ell}_i, t_0) )}{\int \exp( -\tilde{S}_i^{\varepsilon, \lambda_i}(\bar{x}_i^{(0)}, \bar{\ell}_i, t_0) ) \ d\bar{\ell}_i },
		\end{equation}
		and the initial control vector is
		\begin{equation}\label{IniCtrl}
		\tilde{u}_i(\bar{x}_i^{(0)}, \bar{\ell}_i, t_0)  = -\frac{\varepsilon}{\lambda_i}\nabla_{\bar{x}^{(0)}_{i(d)}}q_i(\bar{x}^{(0)}_i, t_0) + (H_i^{(0)})^{-1} \alpha_i^{(0)}.
		\end{equation}
	\end{proposition}
	\begin{proof}
		See \hyperref[App2]{Appendix II} for the proof.
	\end{proof}
	
	We then approximate the PIC~\eqref{OptCtrlPath} and optimal path distribution~\eqref{OptPathDist} with MC method. Given a batch of uncontrolled trajectories $\mathcal{Y}_i = \{ (\bar{x}_i^{(0)}, \bar{\ell}_i^{[y]})  \}_{y = 1, \cdots, Y}$, we can estimate the optimal path distribution~\eqref{OptPathDist} with the following sampling estimator
	\begin{equation}\label{Dist_Estimate}
	\tilde{p}^*_i(\bar{\ell}_i^{[y]} | \bar{x}_i^{(0)}, t_0) \approx \frac{\exp( -\tilde{S}_i^{\varepsilon, \lambda_i}(\bar{x}_i^{(0)}, \bar{\ell}_i^{[y]}, t_0) )}{\sum_{y=1}^{Y} \exp( -\tilde{S}_i^{\varepsilon, \lambda_i}(\bar{x}_i^{(0)}, \bar{\ell}^{[y]}_i, t_0) ) },
	\end{equation}
	where $\tilde{S}_i^{\varepsilon, \lambda_i}(\bar{x}_i^{(0)}, \bar{\ell}_i^{[y]}, t_0)$ denotes the generalized path value of sampled trajectory $(\bar{x}_i^{(0)}, \bar{\ell}_i^{[y]})$. Hence, an estimator for joint optimal control action~\eqref{OptCtrlPath} can be
	\begin{equation}\label{Ctrl_Estimate}
	\bar{u}^*_i= \bar\sigma_i \bar\sigma_i^\top \bar{B}_{i(d)}(x_i)^\top  \sum_{y=1}^{Y} \tilde{p}^*_i(\bar{\ell}^{[y]}_i | \bar{x}_i^{(0)}, t_0)  \tilde{u}_i(\bar{x}_i^{(0)}, \bar{\ell}^{[y]}_i, t_0),
	\end{equation}
	where $\tilde{u}_i(\bar{x}_i^{(0)}, \bar{\ell}^{[y]}_i, t_0)$ is the initial control vector of sampled trajectory $(\bar{x}_i^{(0)}, \bar{\ell}_i^{[y]})$. For a single agent, the sampling procedure can be expedited by the parallel computation of GPU units~\cite{Williams_JGCD_2017}. Meanwhile, instead of generating the joint trajectory set $\mathcal{Y}_i$ from a single agent, we can exploit the computation resource of MAS by letting each agent to sample its local trajectories $\{ (x_i^{(0)}, \ell_i^{[y]})  \}_{y = 1, \cdots, Y}$ and assembling the joint trajectory set $\mathcal{Y}_i$ via communication. With the estimation of joint optimal control action from~\eqref{Ctrl_Estimate}, central agent $i$ acts by following the local control action $u^*_i(\bar{x}_i, t)$ extracted from $\bar{u}^*_i(\bar{x}_i, t)$. \hyperref[alg1]{Algorithm~1} summaries the procedures of cooperative PIC in stochastic MASs.
	\SetAlCapNameFnt{\small}
	\SetAlCapFnt{\small}
	\setlength{\textfloatsep}{1em}% Remove \textfloatsep
	\setlength{\algoheightrule}{0.75pt} % thickness of the rules above and below
	\setlength{\algotitleheightrule}{0.3pt} % thicknes of the rule below the title
	\begin{algorithm}\label{alg1}\small
		\caption{\small Cooperative path integral control algorithm}
		\SetAlgoLined
		\KwIn{agent set $\mathcal{D}$, communication network $\mathcal{G}$, initial time $t_0$, exit time $t_f$, initial states $x_i^{t_0}$, exit states $x_i^{t_f}$, joint state-related costs $q_i(\bar{x}_i)$, control weight matrices $\bar{R}_i$, and exit costs $\phi(x_t^{t_f})$.}
		\renewcommand{\KwResult}{\textbf{Initialization: }} 
		\KwResult{factorial subsystems $\bar{\mathcal{N}}_i$}.\\
		\textbf{Planning \& Execution: }\\
		\For{$t < t_f$ \rm{or} $\bar{x}_i \notin \bar{\mathcal{B}}_i$}{
			\For{$i \in \mathcal{D} = \{1, \cdots, N\}$}{
				Measure joint state $\bar{x}_i(t)$ by collecting state information from neighboring agents $j \in \mathcal{N}_i$\;
				Generate uncontrolled trajectory set $\mathcal{Y}_i$ by sampling or collecting data from neighboring agents\;
				Evaluate generalized path value $\tilde{S}_i^{\varepsilon,\lambda_i}(\bar{x}_i^{(0)}, \bar{\ell}^{[y]}_i,  t_0)$ and initial control $\tilde{u}_i(\bar{x}_i^{(0)}, \bar{\ell}^{[y]}_i,  t_0)$ of each sampled trajectory $(\bar{x}_i^{(0)}, \bar{\ell}_i^{[y]})$ in $\mathcal{Y}_i$ by~\eqref{Gen_Path_Value} and~\eqref{IniCtrl}\;
				Estimate the optimal path distribution $\tilde{p}^*_i(\bar{\ell}_i^{[y]} | \bar{x}_i^{(0)}, t_0)$  and joint optimal control action $\bar{u}^*_i(\bar{x}_i, t)$ by~\eqref{Dist_Estimate} and \eqref{Ctrl_Estimate}\;
				Extract and execute local control action $u_i^*(\bar{x}_i, t)$ from joint optimal control action $\bar{u}^*_i(\bar{x}_i, t)$\;
			}
		} 
	\end{algorithm}

	\section{Simulation Example}\label{sec4}
	We demonstrate the cooperative PIC algorithm in \hyperref[alg1]{Algorithm~1} with a team of cooperative UAVs. Cooperation among UAVs is essential for the tasks that cannot be accomplished by a single UAV and demand multiple UAVs, such as information communication, lifting and carrying heavy load, and patrol with synthetic sensors, which require cooperative UAVs to fulfill some certain constraints, \textit{e.g.} flying closely or maintaining an identical orientation or speed. First, we consider a UAV team with three agents and subject to the communication network in \hyperref[fig2]{Figure~2}.
	\begin{figure}[htpb]\label{fig2} \vspace{-0.5em}
		\centerline{\includegraphics[width=0.4\columnwidth]{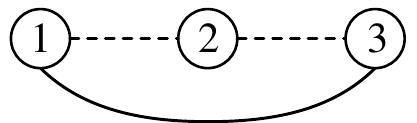}}
		\caption[no botld]{Communication network of a cooperative UAV team with $3$ agents.}
	\end{figure}\vspace{-0.5em}
	UAV 1 and 3, subject to their correlated running cost functions respectively, are tightly connected such that they can cooperate with each other while flying towards their destinations. By contrast, UAV 2, subject to an independent running cost function and only coupled with UAV 1 and 3 via their terminal costs $\phi_i(\bar{x}_i^{t_f})$ as in~\cite{Broek_2008_JAIR}, is loosely connected with other UAVs and will fly to its destination independently. Each UAV is described by the following UAV dynamics~\cite{Broek_2008_JAIR, Yao_CDC_2016}:
	\begin{equation}\label{UAV_Model}
	\setlength\arraycolsep{1pt}
	{\scalefont{0.7} \left(\begin{matrix}
		dx_i\\
		dy_i\\
		dv_i\\
		d\varphi_i
		\end{matrix}\right) = 
		\left(\begin{matrix}
		v_i \cos \varphi_i\\
		v_i \sin \varphi_i\\
		0\\
		0
		\end{matrix}\right) dt +  \begin{pmatrix}
		0 & 0\\
		0 & 0\\
		1 & 0\\
		0 & 1
		\end{pmatrix} \left[ \left( \begin{matrix}
		u_i\\
		\omega_i
		\end{matrix}  \right) dt + \begin{pmatrix}
		\sigma_i & 0\\
		0 & \nu_i
		\end{pmatrix}dw_i
		\right], }
	\end{equation}
	where $(x_i, y_i)$, $v_i$, and $\varphi_i$ denote the position, forward velocity, and heading angle of the $i$-th UAV, respectively; forward acceleration $u_i$ and angular velocity $\omega_i$ are the control inputs, and disturbance $w_i$ is a standard Brownian motion. Control matrix $\bar{B}_i(x_i)$ is constant in~\eqref{UAV_Model}, and we set noise level parameters $\sigma_i = 0.1$ and $\nu_i = 0.05$ in simulation. In order to achieve the requirements that UAV 1 and 3 fly closely towards their destination, and UAV 2 independently flies to its destination, we design the following state-related cost functions
	\begin{align}\label{UAV_RunningCost}
	q_1(\bar{x}_1)  & = w_{11} \cdot (\| (x_1, y_1) - (x_1^{t_f}, y_1^{t_f}) \|_2 - d^{\max}_1) \nonumber \allowdisplaybreaks\\
	& \hspace{35pt} + w_{13} \cdot ( \| (x_1, y_1) - (x_3, y_3) \|_2 - d^{\max}_{13}),\nonumber \allowdisplaybreaks\\
	q_2(\bar{x}_2) & = w_{22}  \cdot (\| (x_2, y_2) - (x_2^{t_f}, y_2^{t_f}) \|_2 - d^{\max}_2),\\
	q_3(\bar{x}_3) & = w_{33} \cdot (\| (x_3, y_3) - (x_3^{t_f}, y_3^{t_f}) \|_2 - d^{\max}_3)  \nonumber\\
	& \hspace{35pt} + w_{31} \cdot ( \| (x_3, y_3) - (x_1, y_1) \|_2 - d^{\max}_{31}),\nonumber
	\end{align}
	where $w_{ii}$ is the weight that contributes to driving agent $i$ towards its exit state $x_i^{t_f}$; $w_{ij}$ is the weight related to the distance between UAVs $i$ and $j$, and $d^{\max}_i$ and $d_{ij}^{\max}$ are the regularization terms for numerical stability, which are respectively assigned by the initial or maximal values of $\| (x_i, y_i) - (x_i^{t_f}, y_i^{t_f}) \|_2$ and $\| (x_i, y_i) - (x_j, y_j) \|_2$.

	In order to verify the improvements brought by cooperative or joint cost functions, we demonstrate an identical flight task while alternating the value of $w_{ij}$ in~\eqref{UAV_RunningCost}. When $w_{ij} > 0$, UAVs 1 and 3 cooperate with each other by flying closely together, while when $w_{ij} = 0$, we restore the factorial running cost functions considered in~\cite{Broek_2008_JAIR}, and UAVs are only correlated via their terminal cost functions $\phi_i(\bar{x}_i^{t_f})$. For three UAVs with initial states $x_1^{0} = (5, 5, 0.3, 0)^\top$, $x_2^{0} = (5, 20, 0.3, 0)^\top$ and $x_3^0 = (5, 35, 0.3, 0)^\top$, we want them to arrive at an identical terminal state $x_i^{t_f} = (35, 20, 0, 0)^\top$ in $t_f = 18$ sec. The period of each control cycle is $0.2$ sec. When generating the trajectory roll-outs $\mathcal{Y}_i$, the time interval from $t$ to $t_f$ is divided into $K = 8$ intervals of equal length $\varepsilon$, \textit{i.e.} $\varepsilon K  = t_f - t$, until $\varepsilon$ becomes less than $0.2$ sec. The noise level parameters are increased to $\sigma_i = 0.75$ and $\nu_i = 0.65$ to improve the sampling and exploration efficiency. The size of $\mathcal{Y}_i$ in estimator~\eqref{Ctrl_Estimate} is $400$ sampled trajectories for each control cycle. Control matrices $\bar{R}_i$ are chosen as identity matrices. The trajectories of UAVs and the relative distance between UAVs 1 and 3 subject to different cost functions are respectively presented in~\hyperref[fig3]{Figure~3} and ~\hyperref[fig4]{Figure~4}. The significant reduction of distance between UAVs 1 and 3 in \hyperref[fig3]{Figure~3}, and \hyperref[fig4]{Figure~4} corroborates that our cooperative PIC algorithm facilitates the interaction and cooperation among the neighboring agents of MASs. For other types of cooperation, such as maintaining an identical orientation, one can realize them by accordingly designing the state-related cost functions. To demonstrate that our cooperative PIC is able to address more complicated tasks, such as obstacle avoidance and multiple objectives, we introduce obstacles (shaded regions) and assign different boundary states to agents in \hyperref[fig5]{Figure~5}.
	\begin{figure}[htpb!]\label{fig3}\vspace{-1em}
		\centerline{\includegraphics[width=0.88\columnwidth]{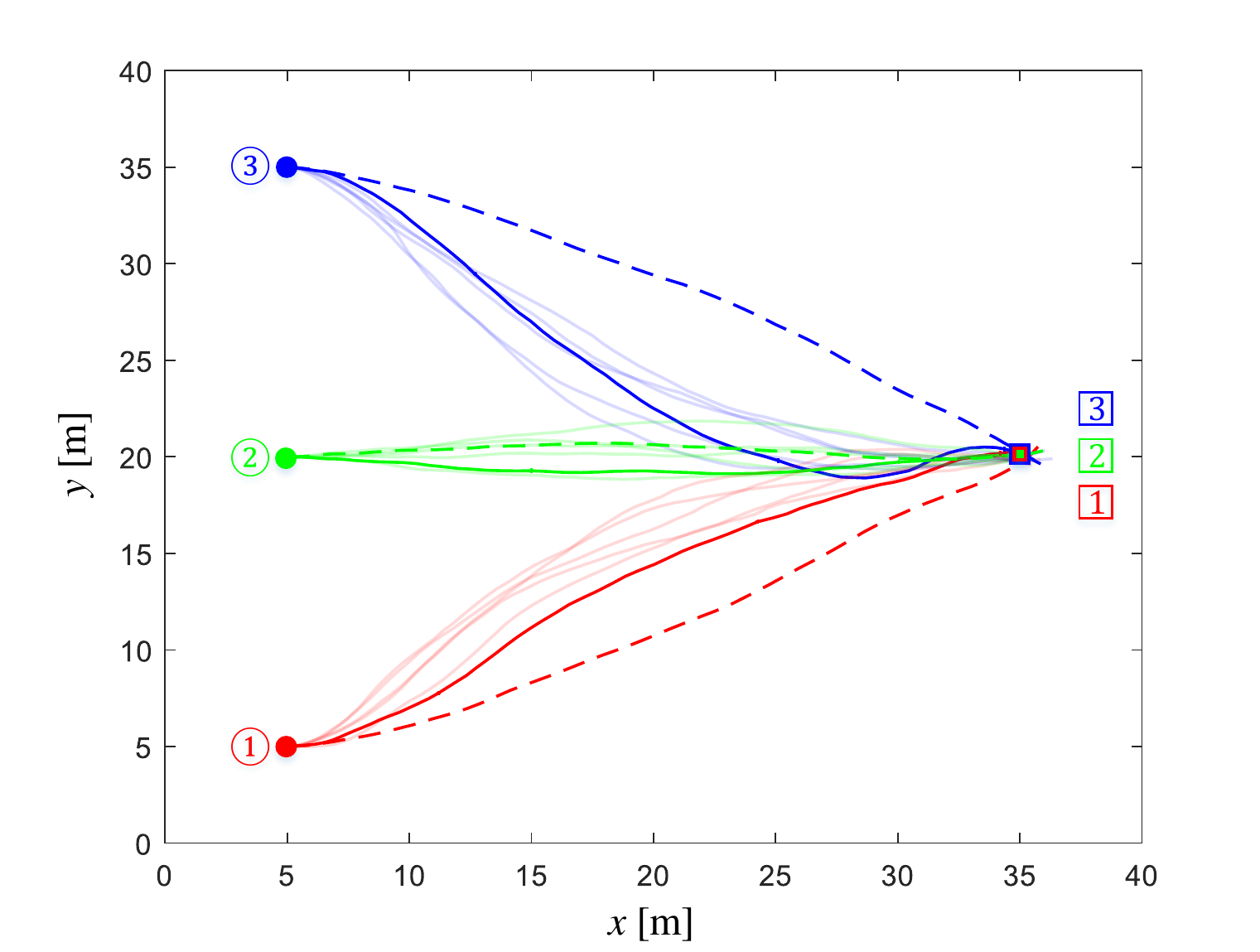}}\vspace{-0.5em}
		\caption{UAV trajectories subject to joint and independent running cost functions. Solid dots and squares are respectively the initial positions and destinations of UAVs. Red, green and blue lines respectively denote the trajectories of UAVs 1, 2 and 3. The solid (transparent) lines are the UAV trajectories subject to joint running cost functions when $w_{11} = w_{33} = 0.7, w_{13} = w_{31} = 1.4$, and $w_{22} = 0.9$ in~\eqref{UAV_RunningCost}. The dashed lines are the UAV trajectories subject to factorial running cost functions when $w_{11} = w_{33} = 0.7, w_{13} = w_{31} = 0$, and $w_{22} = 0.9$ in~\eqref{UAV_RunningCost}.}
	\end{figure}
	\begin{figure}[htpb!]\label{fig4}\vspace{-2em}
		\centerline{\includegraphics[width=0.88\columnwidth]{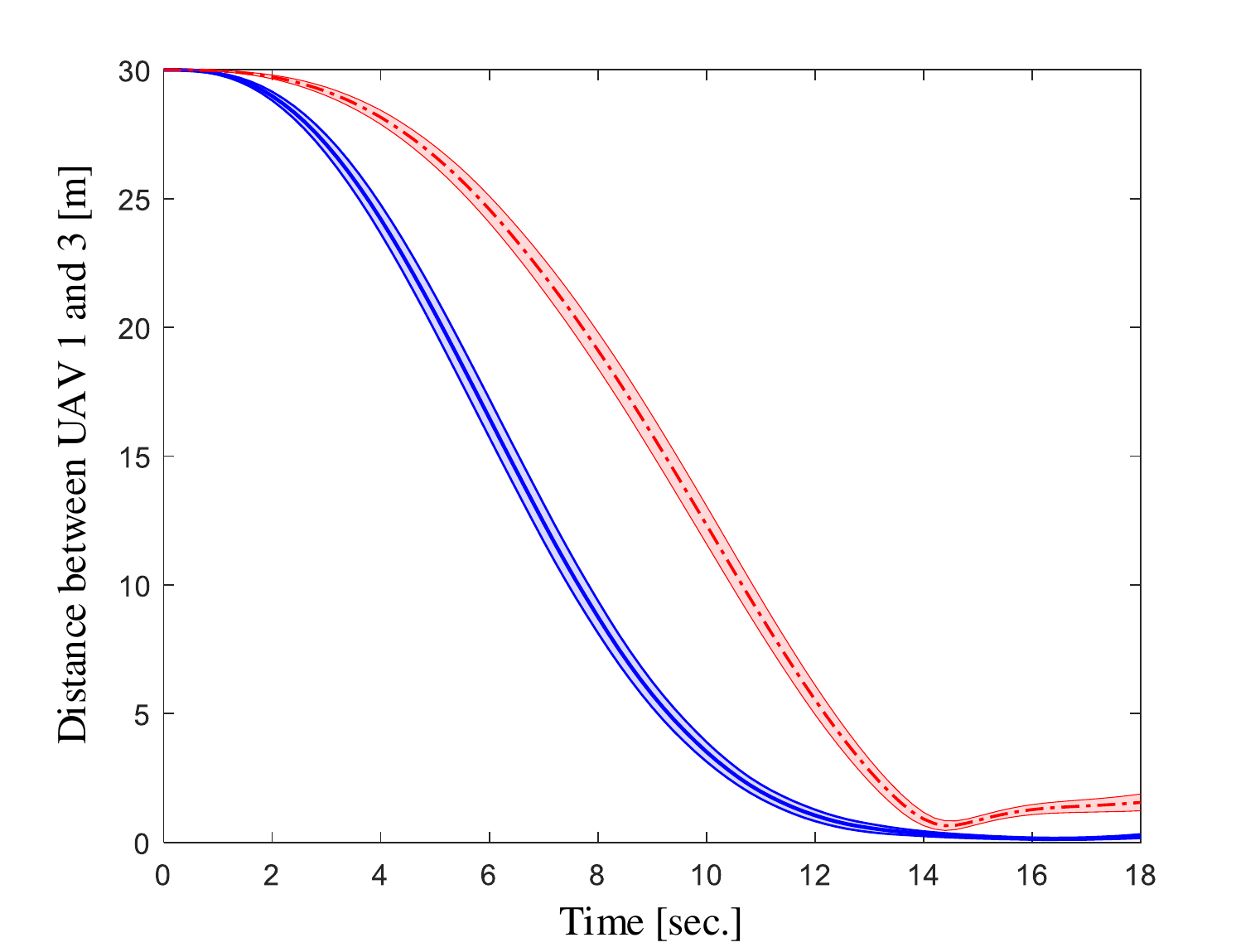}}\vspace{-0.5em}
		\caption{Relative distances between UAVs 1 and 3 subject to joint and independent running cost functions from 100 trails. Blue solid line and red dashed line are the mean distances between UAVs 1 and 3 subject to joint state-related cost with $w_{11} = w_{33} = 0.7, w_{22} = 0.9, w_{13} = w_{31} = 1.4$ and independent state-related cost with $w_{11} = w_{33} = 0.7, w_{22} = 0, w_{13} = w_{31} = 1.4$, respectively. The heights of strips represent one standard deviation of 100 trails.}
	\end{figure}
	\begin{figure}[htpb]\label{fig5}\vspace{-1em}
		\centerline{\includegraphics[width=0.88\columnwidth]{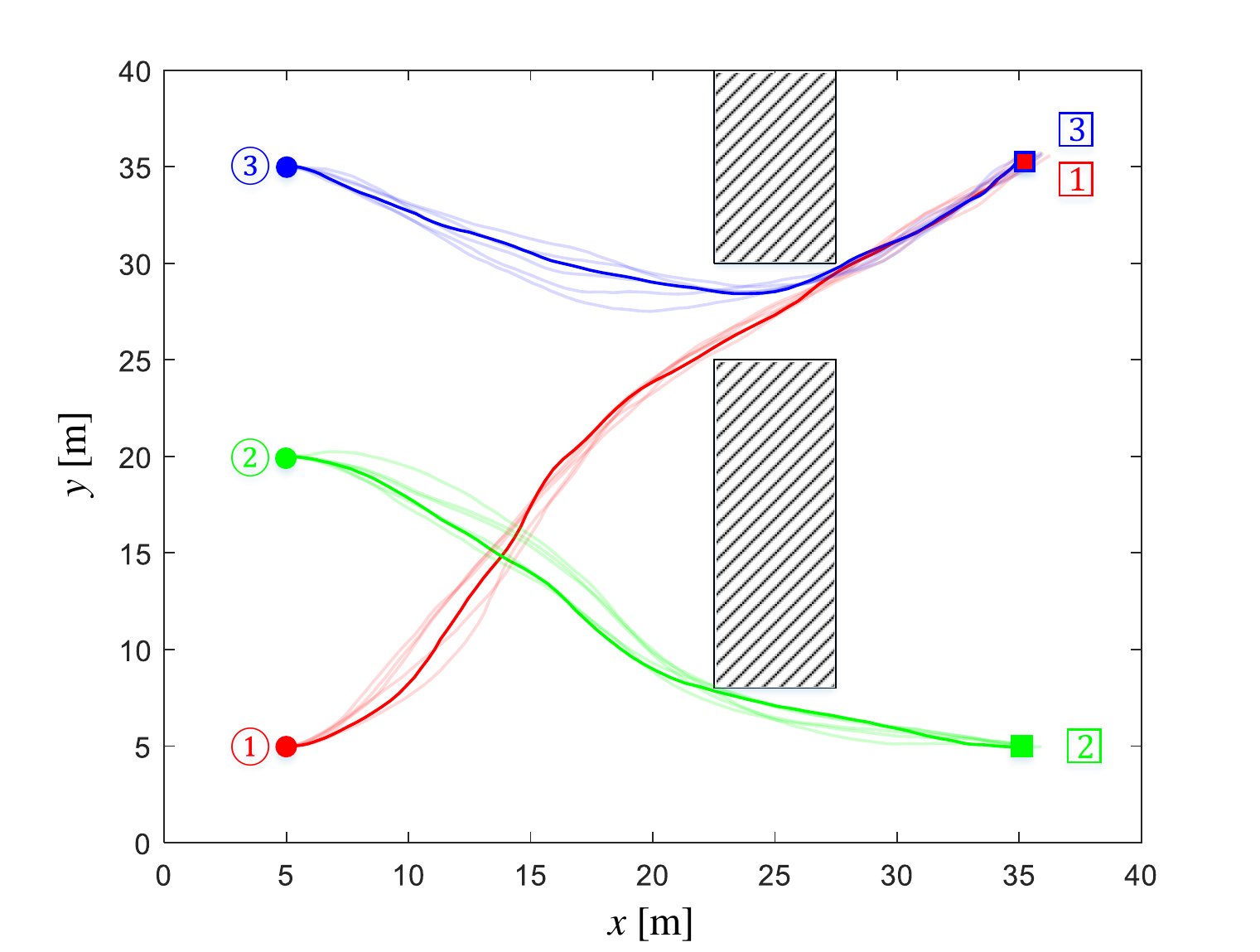}}\vspace{-0.5em}
		\caption{UAV trajectories subject to obstacle avoidance and multiple boundary states. The shaded areas represent obstacles, and the implications of other symbols are the same as in \hyperref[fig3]{Figure~3}. The coefficients of state-related cost~\eqref{UAV_RunningCost} are identical to those in \hyperref[fig3]{Figure~3}, except that a large penalty, \textit{e.g.} $q_i(\bar{x}_i) = 120$, is assigned when central agent $i$ is inside the shaded regions.}
	\end{figure}
	
	 In the end, we test our cooperative PIC scheme on a larger UAV network with 9 agents as shown in \hyperref[fig6]{Figure~6}. For UAV $i \in \{1, \cdots, 9\}$, the initial state is at $x_i^{0} = (10, 100-10\cdot i, 0.5, 0)$, and the state-related cost function is defined by
		\begin{align}
		q_i(\bar{x}_i) = & w_{ii} \cdot (\|(x_i, y_i) - (x_i^{t_f}, y_i^{t_f})\|_2 - d_i^{\max})  \allowdisplaybreaks \\
		&\hspace{5pt}+w_{i, i-1} \cdot (\|(x_i, y_i) - (x_{i-1}, y_{i-1})\| - d^{\max}_{i, i-1}) \nonumber \allowdisplaybreaks\\
		&\hspace{5pt}+ w_{i, i+1} \cdot (\|(x_i, y_i) - (x_{i+1}, y_{i+1})\| - d^{\max}_{i, i+1}) \nonumber
		\end{align}
		where $w_{i, j}$ is the weight related to the distance between agent $i$ and $j$; $w_{i, j} = 0$ when $j = 0$ or $10$; 
		$d_{i,j}^{\max}$ is the regularization term for numerical stability, which is assigned by the initial distance between agent $i$ and $j$ in this example, and the rest of notations are the same as in~\eqref{UAV_RunningCost}. Agents 1 to 6, which share an identical terminal state $A$ at $x_i^{t_f} = (90, 65, 0, 0)^\top$ with $t_f = 40$ sec, are tightly connected via their correlated running cost functions with $w_{ii} = 0.5$, $w_{1,2} =w_{2,3} = w_{3,4} = w_{6, 5} = 1$, $w_{2, 1} = w_{3, 2} = 0$ and $w_{4, 3} = w_{4, 5} = w_{5, 4} = w_{5, 6} = 0.5$. The connection between agents 6 and 7 is loose with $w_{6, 7} = w_{7, 6} = 0$. Agents 7 and 8, which have an identical exit state $B$ at $x_i^{t_f} = (90, 25, 0, 0)$, are tightly correlated with $w_{77} = w_{88} = 0.5$ and $w_{7, 8} = w_{8, 7} = 1$. Agent 9 is loosely connected with agent 8 with $w_{99} = 1$, $w_{9, 8} = 0$, and exit state $C$ at $x_{9}^{t_f} = (90, 10, 0, 0)$. Other parameters are the same as in the first example. A simulation result is presented in \hyperref[fig7]{Figure~7}. For some network topology, \textit{e.g.} loop as in \hyperref[fig2]{Figure~2}, complete binary tree or line as in \hyperref[fig6]{Figure~6}, 
		in which the size of every factorial subsystem is tractable, increasing the total amount of agents in network will not significantly escalate the computation burden on each agent for this cooperative PIC algorithm. 
	\begin{figure}[htpb]\label{fig6}
		\centerline{\includegraphics[width=0.85\columnwidth]{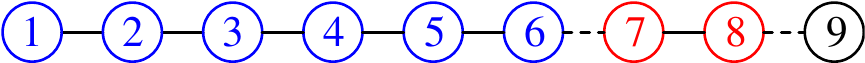}}\vspace{-0.5em}
		\caption{Communication network of a cooperative UAV team with 9 agents.}
	\end{figure}
	\begin{figure}[htpb]\label{fig7} \vspace{-2em}
		\centerline{\includegraphics[width=0.88\columnwidth]{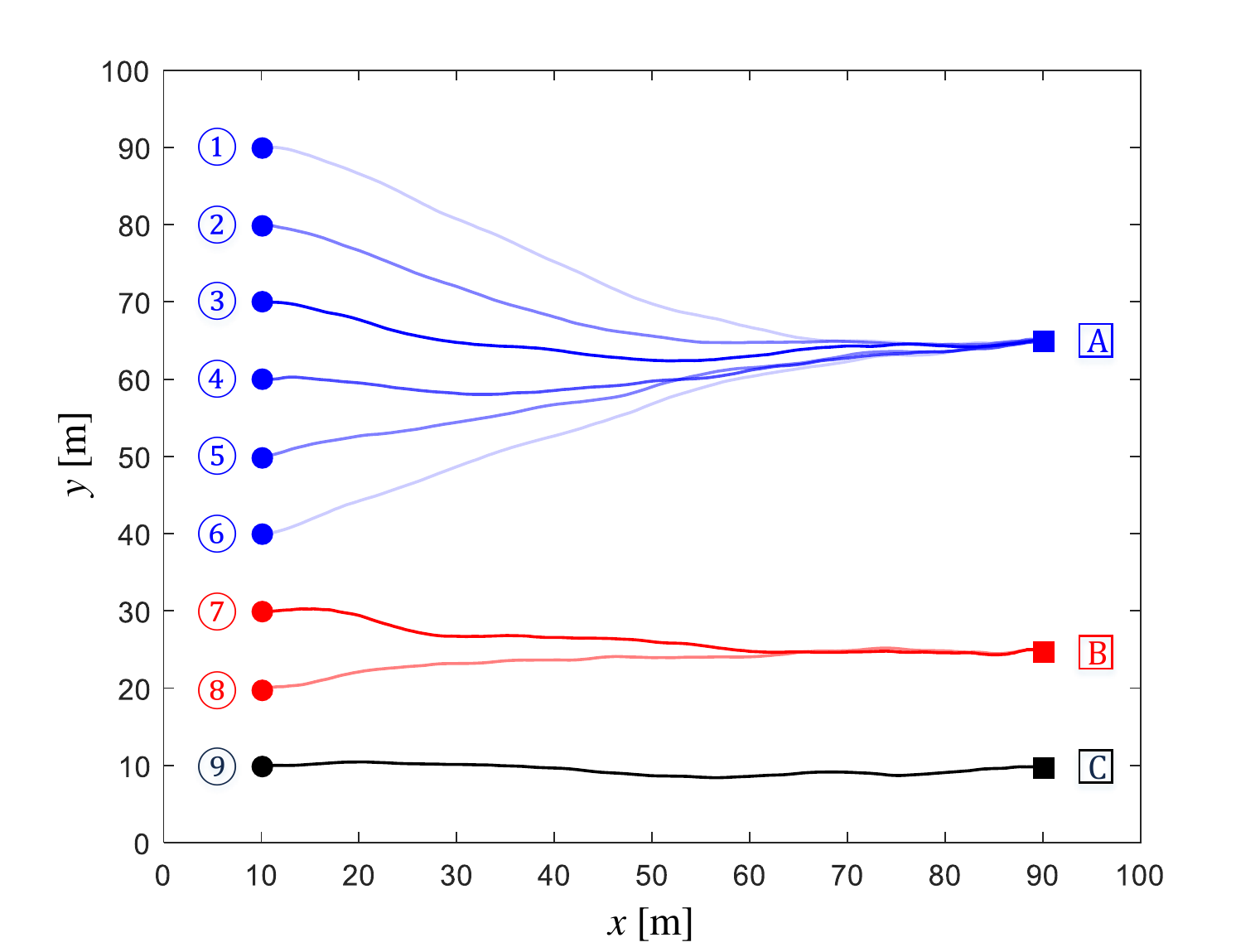}}\vspace{-0.5em}
		\caption{Trajectories of a cooperative UAV team subject to correlated and independent running costs functions.}
	\end{figure}

	\vspace{-1em}
		
	\section{Conclusion}\label{sec5}
	A cooperative path integral control algorithm for stochastic MASs has been investigated in this paper. A distributed control framework that relies on the local observations of agents and hence circumvents the \textit{curse of dimensionality} when the number of agents increases has been proposed, and a cooperative path integral algorithm has been designed to guide each agent to cooperate with its neighboring agents and behave optimally to minimize their joint cost function. Simulation examples have been presented to verify the algorithm.

	\bibliographystyle{IEEEtran}
	% argument is your BibTeX string definitions and bibliography database(s)
	\bibliography{Myref}
	
	\section*{Appendix I: Proof of Theorem 1}\label{AppI}
	We first show that the joint optimality equation~\eqref{optimal_eq} is equivalent to the joint HJB equation~\eqref{sto_HJB} and derive the joint optimal control $\bar{u}_i^*$. Substituting~\eqref{Run_Cost} into~\eqref{optimal_eq} and letting $s$ be a time step between initial time $t$ and exit time $t_f$, the joint equation~\eqref{optimal_eq} can be rewritten as 
	\begin{align}\label{App1_Eq1}
	&V_i(\bar{x}_i, t) = \min_{\bar{u}_i} \mathbb{E}_{\bar{x}_i,t}^{\bar{u}_i} \bigg[ V_i(\bar{x}_i, s)  +  \int_{t}^{s} q_i(\bar{x}_i, \tau)   \\
	&\hspace{90pt} + \frac{1}{2}  \bar{u}_i(\bar{x}_i, \tau)^\top \bar{R}_i\bar{u}_i(\bar{x}_i, \tau) \ d\tau \bigg]. \nonumber
	\end{align}
	Dividing both sides of~\eqref{App1_Eq1} by $s-t$ and letting $s\rightarrow t$, we have
	\begin{align}\label{App1_Eq2}
	0 = \min_{\bar{u}_i} \mathbb{E}_{\bar{x}_i,t}^{\bar{u}_i} \bigg[ \frac{dV_i(\bar{x}_i, t)}{dt} & + q_i(\bar{x}_i, t)\\
	& +  \frac{1}{2} \bar{u}_i(\bar{x}_i, t)^\top\bar{R}_i\bar{u}_i(\bar{x}_i, t) \bigg]. \nonumber
	\end{align}
	Applying It\^o's lemma~\cite{LeGall_2016} to $dV(\bar{x}_i, t)$, dividing the result by $dt$, and taking the expectation over all possible trajectories starting at $(\bar{x}_i^t, t)$ and subject to $\bar{u}_i$, we arrive at
	\begin{align}\label{App1_Eq3}
	&\mathbb{E}_{\bar{x}_i,t}^{\bar{u}_i}\left[ \frac{dV_i(\bar{x}_i, t)}{dt} \right] = \frac{\partial V_i(\bar{x}_i, t)}{\partial t} \allowdisplaybreaks\\
	&\hspace{20pt} + \sum_{j \in \bar{\mathcal{N}}_i}  [f_j(x_j,t)  + B_j(x_j) u_j(\bar{x}_i, t) ]^\top \cdot \nabla_{x_j} V_i(\bar{x}_i,t) \nonumber \allowdisplaybreaks \\
	&\hspace{20pt} + \frac{1}{2} \sum_{j \in \bar{\mathcal{N}}_i}  \textrm{tr} \left( B_j(x_j)\sigma_j \sigma_j^\top B_j(x_j)^\top \cdot \nabla_{x_jx_j}V_i(\bar{x}_i, t) \right), \nonumber
	\end{align}
	where operators $\nabla_{x_i}$ and $\nabla^2_{x_i x_i}$ respectively refer to the gradient and Hessian matrix. Substituting~\eqref{App1_Eq3} into~\eqref{App1_Eq2}, the joint stochastic HJB equation~\eqref{sto_HJB} in \hyperref[thm2]{Theorem~2} is obtained. Taking the derivative of~\eqref{sto_HJB} w.r.t. $\bar{u}_i$ and setting the result to zero, we can obtain the joint optimal control action $\bar{u}_i^*(\bar{x}_i, t)$ in~\eqref{OptimalAction}. 
	
	We then formulate the joint HJB equation~\eqref{sto_HJB} as a linear PDE by the Cole-Hopf transformation~\eqref{CH_Trans}. Subject to \eqref{CH_Trans}, the agent-wise terms in~\eqref{sto_HJB} can be rewritten as
	\begin{equation} \label{App1_Eq4}{\scalefont{0.95} 
		\begin{split}
		&[B_j(x_j)  u_j(\bar{x}_i, t)]^\top \nabla_{x_j}V_i(\bar{x}_i, t)  + \frac{1}{2}{u}_{j}(\bar{x}_i,t)^\top {R}_{j}  u_j(\bar{x}_i, t) =\\
		&\hspace{6pt}\frac{-\lambda_i^2}{2  Z_i(\bar{x}_i, t)^2}  \nabla_{x_j} Z_i(\bar{x}_i,t)^\top \cdot B_j(x_j)R_j^{-1} B_j(x_j)^\top  \nabla_{x_j} Z_i(\bar{x}_i,t),
		\end{split}}
	\end{equation}
	\begin{equation}
	\begin{split}
	&\frac{1}{2} \textrm{tr}   \big(  B_j(x_j)\sigma_j  \sigma_j^\top  B_j(x_j)^\top   \nabla_{x_jx_j} V_i(\bar{x}_i, t)  \big) =  \frac{\lambda_i}{2Z_i(\bar{x}_i,t)^2} \label{App1_Eq5} \\ 
	&\hspace{8pt}\times\textrm{tr} \left( B_j(x_j)\sigma_j \sigma_j^\top B_j(x_j)^\top  \nabla_{x_j}Z_i(\bar{x}_i,t)  \nabla_{x_j}Z_i(\bar{x}_i,t)^\top  \right)\\
	&\hspace{20pt}+\frac{\lambda_i}{2Z_i(\bar{x}_i,t)} \textrm{tr}\left( B_j(x_j)\sigma_j \sigma_j^\top B_j(x_j)^\top  \nabla_{x_jx_j} Z_i(\bar{x}_i,t) \right).
	\end{split}
	\end{equation}
	With identity $R_i = (\sigma_i\sigma_i^\top / \lambda_i)^{-1}$ or $\bar{R}_i = (\bar{\sigma}_i\bar{\sigma}_i^\top / \allowbreak \lambda_i)^{-1}$, the quadratic terms in~\eqref{App1_Eq4} and~\eqref{App1_Eq5} can be canceled, which along with~\eqref{CH_Trans} give the linear PDE in~\eqref{Linear_PDE}. Applying  Feynman-Kac lemma~\cite{LeGall_2016} to~\eqref{Linear_PDE}, we can solve~\eqref{Linear_PDE} and obtain the closed-form solution~\eqref{Z_Solution}. This completes the proof.  \hfill\QED

	\section*{Appendix II: Proof of Proposition 2}\label{App2}
	We first formulate the desirability function $Z_i(\bar{x}_i,t)$ as a path integral. For brevity, we will omit some time arguments $t_k$ or $t_{k+1}$ in this proof. After partitioning the time interval $[t, t_f)$ into $K$ even-length intervals, we can rewrite the expectation in~\eqref{Z_Solution} as
	\begin{align}\label{App2_Eq1}
	Z_i(\bar{x}_i,t) &= \int d\bar{x}^{(1)}_i \cdots \int  \exp\left( -\frac{1}{\lambda_i} \phi_i(\bar{x}^{(K)}_i) \right)  \\
	&\hspace{30pt} \times \prod_{k=0}^{K-1} Z_i(\bar{x}_i^{(k+1)}, t_{k+1}; \bar{x}_i^{(k)}, t_k) \ d\bar{x}_i^{(K)},\nonumber
	\end{align}
	where the function $Z_i(\bar{x}_i^{(k+1)}, t_{k+1}; \bar{x}_i^{(k)}, t_k)$ is implicitly defined by $\int f(\bar{x}_{i}^{(k+1)}) \cdot Z_i(\bar{x}_i^{(k+1)}; \bar{x}_i^{(k)})  \ d\bar{x}_i^{(k+1)} = \mathbb{E}_{\bar{x}_i^{(k)}} \big[  f(\bar{x}_i^{(k+1)}) \cdot \exp\big( -\frac{1}{\lambda_i}\int_{t_i}^{t_{i+1}} q_i(\bar{y}_i, \tau) \ d\tau \big)  \big|  \bar{y}_i(t_k) = \bar{x}_i^{(k)} \big]$ for arbitrary functions $f(\bar{x}_i^{(k+1)})$. Expanding the preceding expectation,  in the limit of infinitesimal $\varepsilon$, we can  approximate $Z_i(\bar{x}_i^{(k+1)}; \bar{x}_i^{(k)})$ by
	\begin{equation}\label{App2_Eq2}
	Z_i(\bar{x}_i^{(k+1)}; \bar{x}_i^{(k)}) = p_i(\bar{x}_i^{(k+1)} | \bar{x}_i^{(k)})  \exp\left(-\frac{\varepsilon}{\lambda_i}  q_i(\bar{x}_i^{(k)}, t_k)    \right),
	\end{equation}
	where $p_i(\bar{x}_i^{(k+1)} | \bar{x}_i^{(k)})$ is the passive transition probability from $(\bar{x}_i^{(k)}, t_k)$ to $(\bar{x}_i^{(k+1)}, t_{k+1})$ and satisfies
	\begin{equation}\label{App2_Eq3}
	p_i(\bar{x}_i^{(k+1)}, t_{k+1} | \bar{x}_i^{(k)}, t_k) \propto p_i(\bar{x}_{i(d)}^{(k+1)}, t_{k+1} | \bar{x}_{i}^{(k)}, t_k).
	\end{equation}
	Since the directly actuated part of uncontrolled dynamics~\eqref{dynamics} satisfies $\bar{x}_{i(d)}^{(k+1)} \sim \mathcal{N}(\bar{x}_{i(d)}^{(k)} + \bar{f}_{i(d)}(\bar{x}^{(k)}_i,  t_k)    \varepsilon, \Sigma^{(k)}_i)$, \break where the covariance is given by $\Sigma^{(k)}_i = \varepsilon \bar{B}_{i(d)}(\bar{x}^{(k)}_i)   \bar{\sigma}_i  \bar{\sigma}^\top_i \cdot \bar{B}_{i(d)}(\bar{x}^{(k)}_i)^\top   = \varepsilon \lambda_i   \bar{B}_{i(d)}(\bar{x}^{(k)}_i)  \bar{R}_i^{-1}  \bar{B}_{i(d)}(\bar{x}^{(k)}_i)^\top = \varepsilon H_i^{(k)}$, and $H_i^{(k)} =  \bar{B}_{i(d)}(\bar{x}^{(k)}_i)   \bar{\sigma}_i   \bar{\sigma}^\top  \bar{B}_{i(d)}(\bar{x}^{(k)}_i)^\top = \lambda_i  \bar{B}_{i(d)}(\bar{x}^{(k)}_i) \cdot \break \bar{R}_i^{-1}   \bar{B}_{i(d)}(\bar{x}^{(k)}_i)^\top$, the transition probability $p_i(\bar{x}_{i(d)}^{(k+1)},  t_{k+1} | \bar{x}_i^{(k)}, t_k)$ in~\eqref{App2_Eq3} satisfies 
	\begin{equation}\label{App2_Eq4}
	p_i(\bar{x}_{i(d)}^{(k+1)}| \bar{x}_i^{(k)})   = \left| 2\pi \Sigma_i^{(k)} \right|^{-1/2}  \exp\left(-\frac{\varepsilon}{2} \left\|\alpha_i^{(k)}\right\|_{(H_i^{(k)})^{-1}} \right),
	\end{equation}
	where $\alpha_i^{(k)} = (\bar{x}_{i(d)}^{(k+1)} - \bar{x}_{i(d)}^{(k)}) / \varepsilon -\bar{f}_{i(d)}(\bar{x}^{(k)}_i, t_k)$. Substituting~(\ref{App2_Eq2})-(\ref{App2_Eq4}) into \eqref{App2_Eq1}, we obtain a path integral for the desirability function 
	\begin{equation}\label{App2_Eq5}
	Z_i(\bar{x}_i, t) = \lim_{\varepsilon \downarrow 0}Z_i^{(\varepsilon)}(\bar{x}^{(0)}_i, t_0),
	\end{equation}
	where the discretized desirability function is given by
	\begin{align*}
	&Z^{(\varepsilon)}_i(\bar{x}^{(0)}_i) = \int \exp\Big( -\tilde{S}_i^{\varepsilon, \lambda_i}(\bar{x}_i^{(0)}, \bar{\ell}_i) \\
	& \hspace{110pt} - {K D |\mathcal{\bar N}_i|} / {2}  \cdot \log (2\pi \varepsilon) \Big) d\bar{\ell}_i;
	\end{align*}
	with path variable $\bar \ell_i = (\bar{x}^{(1)}_i, \cdots, \bar{x}^{(K)}_i)$, and $KD|\bar{\mathcal{N}}_i| / 2 \cdot \log (2\pi \varepsilon)$ is a constant related to  numerical stability. The generalized path value is defined as $\tilde{S}_i^{\varepsilon, \lambda_i}(\bar{x}_i^{(0)}, \bar{\ell}_i) = S_i^{\varepsilon, \lambda_i}(\bar{x}_i^{(0)}, \bar{\ell}_i) + \frac{1}{2}\sum_{k=0}^{K-1}  \log |H_i^{(k)}|$ with the path value
	\begin{align}
	S^{\varepsilon, \lambda_i}_i = \frac{\phi_i(\bar{x}^{(K)}_i)}{\lambda_i} + \varepsilon \sum_{k=0}^{K-1} \Bigg[ \frac{q_i(\bar{x}^{(k)}_i)}{\lambda_i}  + \frac{1}{2} \left\|  \alpha_i^{(k)} \right\|^2_{\left(H_i^{(k)}\right)^{-1}}  \Bigg].\nonumber
	\end{align}
	
	We then compute the path integral formula for joint optimal control action $\bar{u}^*_i(\bar{x}_i, t)$. Substituting the Cole-Hopf transformation~\eqref{CH_Trans} into the joint optimal control action~\eqref{OptimalAction}, we have
	\begin{equation}\label{App2_Eq6}
	\bar{u}^*_{i}(\bar{x}_i, t) = \lambda_i \bar{R}_i^{-1} \bar{B}_i(\bar{x}_i)^\top \frac{\nabla_{\bar{x}_i} Z_i(\bar{x}_i,t)}{Z_i(\bar{x}_i,t)}.
	\end{equation}
	Hence, the path integral formula for joint optimal control can be obtained by substituting \eqref{App2_Eq5} into~\eqref{App2_Eq6}
	\begin{equation}\label{App2_Eq7}
	\bar{u}^*_{i} = -\lambda_i \bar{R}_i^{-1} \bar{B}_i(\bar{x}_i)^\top \cdot \lim_{\varepsilon \downarrow 0} \int \tilde{p}^*_i( \bar{\ell}_i | \bar{x}_i^{(0)}, t_0)  \nabla_{\bar{x}^{(0)}_{i(d)}}\tilde{S}_i^{\varepsilon, \lambda_i}\ d\bar{\ell}_i,
	\end{equation}
	where the optimal path distribution is given in~\eqref{OptPathDist}, and we still need to compute the gradient of the generalized path value. Expanding $\tilde{S}_i^{\varepsilon, \lambda_i}(\bar{x}_i^{(0)}, \bar{\ell}_i, t_0)$ inside the gradient, we have
	\begin{align}\label{App2_Eq8}
	&\nabla_{\bar{x}^{(0)}_{i(d)}}   \tilde{S}_i^{\varepsilon, \lambda_i} = \nabla_{\bar{x}^{(0)}_{i(d)}}\Bigg[  \frac{\phi_i(\bar{x}^{(K)}_i)}{\lambda_i} + \frac{\varepsilon}{\lambda_i}\sum_{k=0}^{K-1}q_i(\bar{x}^{(k)}_i, t_k) \nonumber\\
	&\hspace{55pt} + \frac{\varepsilon}{2} \sum_{k=0}^{K-1} \left\| \alpha_i^{(k)} \right\|_{(H_i^{(k)})^{-1}} + \frac{1}{2}\sum_{k=0}^{K-1} \log \left| H_i^{(k)} \right| \Bigg]. 
	\end{align}
	When the terminal cost is a constant, the gradient of the first term in~\eqref{App2_Eq8} is zero. The gradient of the second term in~\eqref{App2_Eq8} is
	\begin{equation}\label{App2_Eq9}
	\nabla_{\bar{x}^{(0)}_{i(d)}}\frac{\varepsilon}{\lambda_i}\sum_{k=0}^{K-1}q_i(\bar{x}^{(k)}_i, t_k)  = \frac{\varepsilon}{\lambda_i}\nabla_{\bar{x}^{(0)}_{i(d)}}q_i(\bar{x}^{(0)}_i, t_0).
	\end{equation}
	The third gradient in~\eqref{App2_Eq8} satisfies
	\begin{align}\label{App2_Eq10}
	&\nabla_{\bar{x}^{(0)}_{i(d)}} \frac{\varepsilon}{2} \sum_{k=0}^{K-1} \left\| \alpha_i^{(k)} \right\|_{(H_i^{(k)})^{-1}} =- (H_i^{(0)})^{-1} \alpha_i^{(0)}       \\ % 
	& \hspace{75pt}  - \varepsilon \left(  H_i^{(0)}\right)^{-1} \cdot \left[ \nabla_{\bar{x}^{(0)}_{i(d)}} \bar{f}_{i(d)}(\bar{x}_i^{(0)}) \right] \alpha_i^{(0)} \nonumber\\
	& \hspace{75pt} + \frac{\varepsilon}{2} \left( \alpha_i^{(0)} \right)^\top \left[ \nabla_{\bar{x}^{(0)}_{i(d)}}  \left(  H_i^{(0)}   \right)^{-1}  \right]    \alpha_i^{(0)}.  \nonumber
	\end{align}
	The fourth gradient in~\eqref{App2_Eq8} is
	\begin{equation}\label{App2_Eq11}
	\nabla_{\bar{x}^{(0)}_{i(d)}}\frac{1}{2}\sum_{k=0}^{K-1} \log \left| H_i^{(k)} \right| = \frac{1}{2} \nabla_{\bar{x}^{(0)}_{i(d)}} \log \left|  H_i^{(0)}  \right|.
	\end{equation}
	Interested readers can refer to~\cite{Theodorou_JMLR_2010, Theodorou_2011} for more detailed deviation steps on~(\ref{App2_Eq9}-\ref{App2_Eq11}). Meanwhile, when computing the integral of~\eqref{App2_Eq10} in~\eqref{App2_Eq7}, we have
	\begin{align}\label{App2_Eq12}
	&\int \varepsilon \tilde{p}^*_i(\bar{\ell}_i | \bar{x}_i^{(0)}, t_0)   \left(H_i^{(0)}\right)^{-1} \left[ \nabla_{\bar{x}^{(0)}_{i(d)}} \bar{f}_{i(d)}(\bar{x}_i^{(0)}) \right] \alpha_i^{(0)} d\ell = 0, \allowdisplaybreaks \nonumber\\
	&\int \frac{\varepsilon}{2} \tilde{p}^*_i(\bar{\ell}_i | \bar{x}_i^{(0)}, t_0)   \left( \alpha_i^{(0)} \right)^\top \left[ \nabla_{\bar{x}^{(0)}_{i(d)}}  \left(  H_i^{(0)}   \right)^{-1}  \right]    \alpha_i^{(0)} d\ell \allowdisplaybreaks \nonumber \\
	&\hspace{140pt} = - \frac{1}{2} \nabla_{\bar{x}^{(0)}_{i(d)}} \log \left|  H_i^{(0)}  \right|.
	\end{align}
	Substituting (\ref{App2_Eq8})-(\ref{App2_Eq12}) into~\eqref{App2_Eq7}, we obtain the path integral formula for joint optimal control action in (\ref{OptCtrlPath}). This completes the proof. \hfill\QED

\end{document}